\documentclass[11pt]{article}

\usepackage[utf8]{inputenc}

\usepackage{amsmath,amsfonts,amsthm,amssymb,color}
\usepackage[usenames,dvipsnames,svgnames,table]{xcolor}
\definecolor{darkgreen}{rgb}{0.0,0,0.9}
\usepackage{mathtools}
\usepackage{authblk}
\usepackage{fullpage}
\usepackage{parskip}
\usepackage{comment}
\usepackage{tikz}
\usepackage{bbm}
\usepackage{dsfont}
\usepackage[sc]{mathpazo}
\usepackage[basic]{complexity}
\usepackage{algorithm2e}
\usepackage[colorlinks=true,
citecolor=OliveGreen,linkcolor=BrickRed,urlcolor=BrickRed,pdfstartview=FitH]{hyperref}
\usepackage[capitalize,nameinlink]{cleveref}
\usepackage{tcolorbox}
\usepackage[short]{optidef}

\newtcolorbox{wbox}
{
	colback  = white,
}

\SetKwInOut{Input}{Input}
\SetKwInOut{Output}{Output}
\SetKwFunction{Uncross}{\textsc{Uncross}}
\SetKwFunction{MergeUncross}{\textsc{MergeUncross}}
\SetKwFunction{PerfectMatching}{\textsc{PerfectMatching}}
\SetKwBlock{InParallel}{in parallel do}{end}
\SetKwFor{ParallelFor}{for}{in parallel do}{end}
\newcommand*{\suppress}[1]{}

\newcommand*{\cM}{\mathcal{M}}

\newcommand*{\cR}{\mathcal{R}}

\newcommand*{\oG}{\overline{G}}
\newcommand*{\oE}{\overline{E}}

\newcommand{\worth}{\mbox{\rm worth}}
\newcommand{\money}{\mbox{\rm money}}

\newcommand{\profit}{\mbox{\rm profit}}

\makeatletter
\def\thm@space@setup{%
	\thm@preskip= 10pt
	\thm@postskip=\thm@preskip % or whatever, if you don't want them to be equal
}
\makeatother

\makeatletter
\renewcommand{\paragraph}{%
	\@startsection{paragraph}{4}%
	{\z@}{5pt}{-1em}%
	{\normalfont\normalsize\bfseries}%
}
\makeatother

\newtheorem{theorem}{Theorem}
\newtheorem{lemma}{Lemma}

\newtheorem{corollary}{Corollary}

\theoremstyle{definition}
\newtheorem{definition}{Definition}

\newtheorem{remark}{Remark}

\newtheorem{example}{Example}

\newenvironment{fminipage}%
{\begin{Sbox}\begin{minipage}}%
		{\end{minipage}\end{Sbox}\fbox{\TheSbox}}

\newcommand{\CM}{\mbox{${\cal M}$}}

\newcommand\QQ{\boldsymbol{\mathit{Q}}}

\newcommand{\cost}{\mbox{\rm cost}}

\newcommand\ZZ{\boldsymbol{\mathit{Z}}}

\title{The Investment Management Game: \\
Extending the Scope of the Notion of Core}

\author[1]{Vijay V.~Vazirani\footnote{Supported in part by NSF grant CCF-2230414.}}

\affil[1]{University of California, Irvine}

\date{}

\begin{document}
	\maketitle

\begin{abstract}
The core is a dominant solution concept in economics and cooperative game theory; it is predominantly used for profit — equivalently, cost or utility — sharing. This paper demonstrates the versatility of this notion by proposing a completely different use: in a so-called investment management game, which is a {\em game against nature} rather than a cooperative game. This game has only one agent whose strategy set is all possible ways of distributing her money among investment firms. 

The agent wants to pick a strategy such that in each of exponentially many future ``scenarios'', sufficient money is available in the “right” firms so she can buy an “optimal investment” for that scenario. Such a strategy constitutes a core imputation under a broad interpretation, though traditional formal framework, of the core.

Our game is defined on {\em perfect graphs}, since the maximum stable set problem can be solved in polynomial time for such graphs. We completely characterize the core of this game, analogous to Shapley and Shubik’s characterization of the core of the assignment game. A key difference is the following technical novelty: whereas their characterization follows from total unimodularity, ours follows from {\em total dual integrality}.
\end{abstract}

\begin{comment}
	The core is a dominant solution concept in economics and game theory. In this context, the following question arises, ``How versatile is this solution concept?'' We note that within game theory, this notion has been used for profit ---  equivalently, cost or utility --- sharing only, in the setting of   cooperative game theory.

In this paper, we show a completely different use for this notion: in an {\em investment management game}  which is not a cooperative game; it is a game against nature. This game has only one agent who needs to allocate her money among investment firms in such a way that {\em in each of exponentially many future scenarios}, sufficient money is available in the ``right'' firms so she can buy an ``optimal investment'' for that scenario. 

Our game is defined on {\em perfect graphs}, since the maximum stable set problem in such graphs is in ${\cal P}$. We completely characterize the core of this game in a way that is analogous to Shapley and Shubik's characterization of the core of the assignment game. A key  difference is the following: whereas their characterization follows from {\em total unimodularity}, ours follows from {\em total dual integrality}. The latter is another novelty of our work.  

\end{comment}

\bigskip
\bigskip
\bigskip
\bigskip
\bigskip
\bigskip
\bigskip
\bigskip
\bigskip
\bigskip
\bigskip
\bigskip
\bigskip
\bigskip
\bigskip
\bigskip
\bigskip
\bigskip
\bigskip
\bigskip
\bigskip
\bigskip

\pagebreak

\section{Introduction}
\label{sec.intro}

The core is a dominant solution concept in economics and game theory. Its origins lie in the 19th century book of Edgeworth \cite{Edgeworth1881mathematical} in which it was referred to as the {\em contract curve} and was first used in general equilibrium theory. In 1959, Gillies \cite{Gillies-Core} gave an improved definition which started being used in cooperative game theory for ``fair'' profit sharing. Since then, the stature of this solution concept grew considerably and today it is considered the gold standard for profit sharing and much more desirable than other notions, e.g., least core and nucleolus\footnote{Some drawbacks of these two notions are mentioned in \cite{Va.general}.}.

The following question arises, ``How versatile is this solution concept?'' Let us first consider this question in the context of Nash equilibrium, which is perhaps the most important solution concept in game theory. Even for the special case of bimatrix games, it provides deep insights in a rich milieu of situations, each having its own special character, e.g., Prisoner's Dilemma, Matching Pennies,  Battle of the Sexes, and Rock-Paper-Scissor; the last game is in fact zero-sum. Indeed, most ``big'' solution concepts tend to be similarly multi-faceted. In contrast, within game theory, the core has been used for profit ---  equivalently, cost or utility --- sharing only. 

The purpose of this paper is to show that the notion of core is more versatile than previously believed. For this, we need to first pick an appropriate definition of core, since adopting a narrow definition would be self-defeating; the latter says that this notion is meant for sharing profit in such a way that stability is ensured. Under this definition, a core imputation shares the profit of the grand coalition among agents so that no sub-coalition can make more profit by itself and therefore will not succeed from the grand coalition. Clearly, we need to interpret core in more general terms; however, we do not wish to take liberties with its formal framework, see Section \ref{sec.prelims}. Fortunately, a middle ground exists, as described in Section \ref{sec.graph}.

Definition \ref{def.game-management} introduces the {\em investment management game},
%\footnote{As stated in Section \ref{sec.discussion}, our first formulation was purely graph-theoretic. The economic interpretation given in Definition \ref{def.game-management} is clearly far from realistic. However, we note that our paper has little to do with immediate applicability --- its purpose is to make a conceptual advance on a fundamental solution concept.}, 
whose core we will characterize. Interestingly, it is not a cooperative game; in fact, it has {\em only one agent}. Our game is best viewed as a {\em game against nature}, i.e., a game against a player whose pay-offs, as well as probability distribution over strategies, are unknown. The origin of this name comes from the experience of farmers from a time when weather prediction and irrigation systems were not well developed, making it difficult for them to predict the best choice of crops and forcing them to play in complete ignorance \cite{Biswas-Nature}. Our game has exactly this character, see  Definition \ref{def.game-management}. The saving grace is that in our game, a core imputation helps salvage the situation by enabling the agent to invest in such a way she is able to respond to {\em every strategy} of nature optimally! 
 
 After defining our game, in Remark \ref{rem.contrast} we draw a clear distinction between the manner in which the core is applied in cooperative game theory and in our game. In Section \ref{sec.graph}, we define our game on a graph. Its main computational problem that arises in it is to find a maximum weight independent set, which is ${\cal NP}$-hard in arbitrary graphs. However, the restriction of this problem to {\em perfect graphs} (Section \ref{sec.perfect}) is in ${\cal P}$  \cite{GLS}; this will be the setting for our game. 

In Section \ref{sec.perfect} we characterize the core of this game. A good way of describing this result is by drawing an analogy with the classic paper of Shapley and Shubik \cite{Shapley1971assignment} which characterized the core of the assignment game. They showed that the core of this game is precisely the set of optimal solutions to the dual of the LP-relaxation of the maximum weight matching problem in the underlying bipartite graph, see Section \ref{sec.prelims}.  As observed in \cite{Deng1999algorithms,Demange-Deng,Va.LP}, at the heart of the Shapley-Shubik proof lies the fact that the polytope defined by the constraints of this LP is integral, see Definition \ref{def.vertices-integral}. In turn, integrality follows from the fact that the constraint matrix of this LP is {\em totally unimodular (TUM)}, see Definition \ref{def.totally-unimodular}. The underlying reason for this  requirement of integrality is an inherent indivisibility in the game, e.g., in a cooperative game, such as the assignment game, agents are indivisible. 

Integrality holds for the maximum weight independent set problem for perfect graphs as well. However, the underlying reason is not TUM, but the more general condition of {\em total dual integrality (TDI)}, see Definition \ref{def.TDI}. Building on this fact, we show that the set of core imputations of our game is precisely the set of optimal solutions to the dual of the LP-relaxation of the maximum weight independent set problem for perfect graphs. Ours appears to be the first work which uses TDI for characterizing the core of a game; previous characterizations were based on TUM, see Section \ref{sec.related}.
 
Another novelty of our work is the following: In cooperative games, the profit of a sub-coalition under an imputation was defined via a {\em bottom-up process} in which the worth of the game was distributed among agents and the profit of a sub-coalition was simply the sum of profits of agents in it. In our game there is only one agent and the natural process is {\em top-down}, as described in Section \ref{sec.graph} and summarized in Remark \ref{rem.top-down}.

\begin{definition}
\label{def.game-management}
The {\em investment management game} involves {\em one agent}, a set $V$ of {\em assets} with a {\em cost function} $w: V \rightarrow \QQ_+$, and a set $\cM$ of {\em investment firms}. For concreteness, assume that assets are shares of specific companies, each worth a specific dollar amount; obviously, the number of shares in an asset changes according to the going price. Assume further that, much like a mutual fund, each investment firm specializes in holding shares of specific types of companies, e.g., Internet companies, software companies, computer hardware companies, automobile companies, etc. The difference is that whereas each share of a mutual fund represents a collection of assets in some predetermined proportions, in our setting, the agent {\em can buy individual assets} from a firm. 

The set of assets sold by a specific firm $f \in \cM$ is denoted by $Q_f \subseteq V$. The same asset, such as ``five thousand dollars worth of Microsoft'', may be sold by more than one firm, e.g., this  asset may be sold by a firm specializing in Internet companies as well as a firm specializing in software companies. Each subset $S \subseteq V$ of assets is called a {\em scenario}; hence there are exponentially many, namely $2^{|V|}$, scenarios.

Two or more assets which are sold by the same firm have obvious correlations and therefore do not constitute a low-risk investment. On the other hand, a set of assets which picks at most one asset from any investment firm constitutes a {\em diversified portfolio} --- it avoids correlations and is therefore considered a ``healthy'' investment. An {\em optimal investment} for a scenario $S$ is a largest possible diversified portfolio in it; formally, it is defined to be a maximum cost set $S' \subseteq S$ which picks at most one asset from any investment firm.

The {\em total money}, $T \in \QQ_+$ of the agent is just sufficient to buy an optimal investment in scenario $V$. The {\em strategy set} of the agent is all possible ways of distributing money $T$ among  the firms; each strategy will also be called an {\em imputation}. The rules of this game dictate that the money allocated to a firm $f \in \cM$ is available for use at {\em any asset} $v \in Q_f$. Note that this money need not be used for buying $v$; it is simply available at $v$. Therefore, the {\em money available in scenario} $S$ is the sum of money allocated to all firms which have at least one asset in $S$. Clearly, our game is a transferable utility (TU) game.

%We will assume the following flexibility in the way assets are bought: the asset(s) which get bought with the money available at $v$ do not have to necessarily include $v$. 

The {\em game} is the following: Find an imputation such that when, at a certain time in the future, nature picks a strategy, i.e., a scenario $S \subseteq V$, the money available in $S$ is sufficient to buy an optimal investment in $S$. Such an imputation is said to be in the {\em core} of the game. Thus a core imputation enables the agent to invest $T$ money in firms in such a way she is able to respond to {\em every strategy} of nature optimally. 
\end{definition}

\begin{remark}
	\label{rem.contrast}
In light of Definition \ref{def.game-management}, a clear distinction can be given between the manner in which the core is applied in cooperative game theory and in our game. Whereas in the former, a core imputation defines payoffs of all players that results in the stability of the grand coalition, in the latter, a core imputation captures an investment strategy, of the unique player, that encapsulates risk aversion (for the stated definition of ``risk'') under every future scenario. 
\end{remark}

\section{Related Works}
\label{sec.related}

Results characterizing cores of natural Transferable Utility (TU) games are given below. First, we mention the stable matching game, defined by Gale and Shapley \cite{GaleS}, in which preferences are ordinal, i.e., it is an NTU game. The only coalitions that matter in this game are ones formed by one agent from each side of the bipartition. A stable matching ensures that no such coalition has the incentive to secede and the set of such matchings constitute the core of this game. Vande Vate \cite{Vate1989linear} and Rothblum \cite{Rothblum1992characterization} gave linear programming formulations for stable matchings; the vertices of their underlying polytopes are integral and are stable matchings. Interestingly enough,  Kiraly and Pap \cite{TDI-Kiraly2008total} showed that the linear system of Rothblum is in fact TDI. 

A core imputation has to ensure that {\em each} of the exponentially many sub-coalitions is ``satisfied'' --- clearly, that is a lot of constraints. As a result, the core is known to be non-empty only for a handful of games, some of which are mentioned below; total unimodularity plays a key role in these results.

Deng et al. \cite{Deng1999algorithms} observed the role of integrality in the Shapley-Shubik Theorem and used this insight to distilled its underlying ideas to obtain a general framework which helps characterize the cores of several games that are based on fundamental combinatorial optimization problems, including maximum flow in unit capacity networks both directed and undirected, maximum number of edge-disjoint $s$-$t$ paths, maximum number of vertex-disjoint $s$-$t$ paths and maximum number of disjoint arborescences rooted at a vertex $r$. 

The survey of Demange and Deng \cite{Demange-Deng} on the notion of balancedness of Bondareva and Shapley \cite{Bondareva1963some, Shapley1965balanced} explored the role of integrality in depth. Towards the end of their paper, they note that TDI of a linear system is a very general condition that leads to integrality and therefore non-emptyness of the core and balancedness. However, they did not give the example of any game for which TUM does not hold and TDI has to be invoked for establishing integrality.

A natural generalization of the assignment game is the $b$-matching game in bipartite graphs. Biro et al. \cite{Biro2012computing} showed that the core non-emptiness and core membership problems for the $b$-matching game are solvable in polynomial time if $b \leq 2$ and are co-NP-hard even for $b = 3$. More recently, Vazirani \cite{Va.Char} showed that if $b$ is the constant function, then core imputations are precisely optimal solutions to the dual LP; this is analogous to the Shapley-Shubik theorem. Furthermore, Vazirani \cite{Va.Char} showed that if $b$ is arbitrary, then every optimal solutions to the dual LP is a core imputations; however, there are core imputations that are not optimal solutions to the dual LP. 

We next describe results for the facility location game. First, Kolen \cite{Kolen-facility} showed that for the unconstrained facility location problem, each optimal solution to the dual of the classical LP-relaxation is a core imputation if and only if this relaxation has no integrality gap. Later, Goemans and Skutella \cite{Goemans-Skutella} showed a similar result for any kind of constrained facility location game. They also proved that in general, for facility locations games, deciding whether the core is non-empty and whether a given allocation is in the core is NP-complete. 

Samet and Zemel \cite{Samet-Zemel} study games which are generated by linear programming optimization problems; these are called LP-games. For such games, It is well known that the set of optimal dual solutions is contained in the core and \cite{Samet-Zemel} gives sufficient conditions under which equality holds. These games do not ask for integral solutions and are therefore different in character from the ones studied in this paper.

Granot and Huberman \cite{Granot1981minimum, Granot-2-1984core} showed that the core of the minimum cost spanning tree game is non-empty and gave an algorithm for finding an imputation in it. Koh and Sanita \cite{Laura-Sanita} settle the question of efficiently determining if a spanning tree game is submodular; the core of such games is always non-empty. Nagamochi et al. \cite{Nagamochi1997complexity} characterize non-emptyness of core for the minimum base game in a matroid; the minimum spanning tree game is a special case of this game.

\section{Definitions and Preliminary Facts}
\label{sec.prelims}

In this section, we will give the standard definition of core in the setting of a cooperative game. For completeness, we will also formally state the characterization of the core of the assignment game given by Shapley and Shubik. In Section \ref{sec.graph} we will modify some of these definitions, and appropriately rename others, so that the notion of core can be used to study the investment management game. 

We will study {\em transferable utility (TU)} games, i.e., a games in which utilities of the agents are stated in monetary terms and side payments are allowed. For an extensive coverage of these notions, see the book by Moulin \cite{Moulin2014cooperative}.

\begin{definition}
	\label{def.cooperative-game}
	A {\em cooperative game} consists of a pair $(N, c)$ where $N$ is a set of $n$ agents and $v$ is the {\em characteristic function}; $c: 2^N \rightarrow \cR_+$, where for $S \subseteq N, \ c(S)$ is the {\em worth} that the sub-coalition $S$ can generate by itself. $N$ is also called the {\em grand coalition}.
\end{definition}

\begin{definition}
	\label{def.imputation}	
	An {\em imputation} gives a way of dividing the worth of the game, $v(N)$, among the agents. It can be viewed as a function $x: {N} \rightarrow \QQ_+$. For each sub-coalition $S \subseteq N$, we will define its {\em profit} as $\profit(S) = \sum_{i \in S} {x(i)}$.   
\end{definition}

\begin{definition}
	\label{def.satisfied}
Let $x$ be an imputation and $S \subseteq N$ a sub-coalition. We will say that $x$ {\em satisfies} $S$ if its profit is at least as large as its worth, i.e., $\profit(S) \geq \worth(S)$. 
\end{definition}
	
\begin{definition}
	\label{def.core}
	An imputation $x$ is said to be in the {\em core of the game} if it satisfies each sub-coalition $S \subseteq N$.
\end{definition}
 
The {\em assignment game}, consists of a bipartite graph $G = (U, V, E)$ and a weight function $w: E \rightarrow \QQ_+$. The agents of this game are $U \cup V$ and for each sub-coalition $(S_u \cup S_v)$, with $S_u \subseteq U$ and  $S_v \subseteq V$, its worth is defined to be the weight of a maximum weight matching in $G(S_u \cup S_v)$, where the latter is the subgraph of $G$ induced on the vertices $(S_u \cup S_v)$.

Linear program (\ref{eq.core-primal-bipartite}) gives the LP-relaxation of the problem of finding such a matching. In this program, variable $x_{ij}$ indicates the extent to which edge $(i, j)$ is picked in the solution. 

	\begin{maxi}
		{} {\sum_{(i, j) \in E}  {w_{ij} x_{ij}}}
			{\label{eq.core-primal-bipartite}}
		{}
		\addConstraint{\sum_{(i, j) \in E} {x_{ij}}}{\leq 1 \quad}{\forall i \in U}
		\addConstraint{\sum_{(i, j) \in E} {x_{ij}}}{\leq 1 }{\forall j \in V}
		\addConstraint{x_{ij}}{\geq 0}{\forall (i, j) \in E}
	\end{maxi}

The constraint matrix of LP (\ref{eq.core-primal-bipartite}) is totally unimodular (TUM), see Definition \ref{def.totally-unimodular}, and therefore the polytope defined by its constraints is integral, e.g., see \cite{LP.book}. Taking $u_i$ and $v_j$ to be the dual variables for the first and second constraints of (\ref{eq.core-primal-bipartite}), we obtain the dual LP: 

 	\begin{mini}
		{} {\sum_{i \in U}  {u_{i}} + \sum_{j \in V} {v_j}} 
			{\label{eq.core-dual-bipartite}}
		{}
		\addConstraint{ u_i + v_j}{ \geq w_{ij} \quad }{\forall (i, j) \in E}
		\addConstraint{u_{i}}{\geq 0}{\forall i \in U}
		\addConstraint{v_{j}}{\geq 0}{\forall j \in V}
	\end{mini}

\begin{theorem}
	\label{thm.SS}
	(Shapley and Shubik \cite{Shapley1971assignment})
	The imputation $(u, v)$ is in the core of the assignment game if and only if it is an optimal solution to the dual LP, (\ref{eq.core-dual-bipartite}). 
\end{theorem}

\begin{definition}
	\label{def.vertices-integral}
	We will say that a {\em polytope is integral} if its vertices have all integral coordinates. 
\end{definition}

\begin{definition}
	\label{def.totally-unimodular}
	Let $Ax \leq b$ be a linear system where $A$ is an $m \times n$ matrix and $b$ an $m$-dimensional vector, both with integral entries.  $A$ is said to be {\em totally unimodular (TUM)} if every submatrix of $A$ has determinant $0, 1$ or $-1$. If so, the polytope of this linear system is  integral. 
\end{definition}

\begin{definition}
	\label{def.TDI}
	Let $Ax \leq b$ be a linear system where $A$ is an $m \times n$ matrix and $b$ an $m$-dimensional vector, both with rational entries. We will say that this linear system is {\em totally dual integral (TDI)} if for any integer-valued vector $c^T$ such that the linear program
	
\[ \max \{ cx: Ax \leq b \} \]

has an optimum solution, the corresponding dual linear program has an {\em integer} optimal solution. If so, the polytope of this linear system is integral. 
\end{definition}

Note that TDI is a more general condition than TUM for integrality of polyhedra. If $A$ is TUM then the polyhedron of the linear system $Ax \leq b$ is integral for every integral vector $b$. However, even if $A$ is not TUM, for specific choices of an integral vector $b$, the polytope of the linear system $Ax \leq b$ may be integral, and TDI may apply in this situation. It is important to remark that TDI is not a property of the polytope but of the particular linear system chosen to define it. See \cite{GLS, Sch-book} for further details.

\section{The Investment Management Game Defined on a Graph}
\label{sec.graph}

As stated in the Introduction, the investment management game, which is described at a high level in Definition \ref{def.game-management}, is a game against nature and it is not a cooperative game; in fact it has only one agent. A natural setting for this game is a graph, as described below. This is a  TU game: the money available at a firm can be used to buy assets from other firms, as specified below.  In order to study its core, we will modify some of the standard definitions given in Section \ref{sec.prelims},  and appropriately rename others.

Let $G = (V, E)$ be a graph whose vertices are {\em assets}; let $|V| = n$. The function $w: V \rightarrow \QQ_+$ defines the {\em cost} of each asset. Every {\em maximal clique}\footnote{It is easy to see that our result will hold even if we had defined every clique to be an investment firm. However, under that formulation, there would be numerous pairs of firms $f, f'$ with $Q_{f'} \subset Q_f$, making firm $f'$ redundant. Restricting to maximal cliques avoids this deficiency in the formulation.} in $G$ is an {\em investment firm}; let $\CM$ denote the set of all firms. For each firm $f \in \cM$, the assets sold by $f$ are represented by the vertices in its clique, $Q_f \subseteq V$. As stated in Definition \ref{def.game-management}, each investment firm specializes in holding shares of a specific type of companies. As a result, two or more assets held by the same firm will have obvious correlations. 

We now explain why investment firms are defined to be cliques, via the following analogy. Consider all the ATMs in the US of a certain bank. Since the money deposited in this bank, or in an ATM of this bank,  can be withdrawn from any of its ATMs, the set of all ATMs of this bank can be viewed as a clique, interconnected via a network. Similarly, since the money allocated to an investment firm is available for use at any of its assets, as stated in Definition \ref{def.game-management}, we have defined the investment firm to be a clique --- over its assets. 

%Maximality helps us avoid redundancy; clearly, an investment firm defined by a sub-clique $Q'$ of a clique $Q$ is redundant.  

Every set $S \subseteq V$ is called a {\em scenario}, i.e., scenario plays the same role as sub-coalition in a cooperative game. Let $G(S)$ denote the subgraph of $G$ induced on vertices in $S$. As required by Definition \ref{def.game-management}, an {\em optimal investment} in scenario $S$ is defined to be any maximum cost independent set\footnote{An independent set is also called a stable set, see Definition \ref{def.stable-set}.} in $G(S)$; clearly, this is a diversified investment since it picks at most one asset from any investment firm. Let $O_S \subseteq S$ denote such an investment and define
$$ \cost(S) := \sum_{v \in O_S} {w_v} .$$ 

Note that $\cost$ plays the same role as $\worth$ in  Definition \ref{def.cooperative-game} and the function $\cost: 2^V \rightarrow \QQ_+$ plays the same role as the characteristic function. The {\em total money} of the agent is defined to be $T = \cost(V)$, i.e., just sufficient to buy an optimal investment in $G$. 
 
A function $y: \cM \rightarrow \QQ_+$ where $\sum_{Q \in \cM} {y_Q} = T$ is called an {\em imputation}, i.e., it is a way of distributing $T$ money among the investment firms. The set of all such functions $y$ can also be viewed as the {\em strategy set} of the unique agent in the game. For any scenario $S$, the {\em money available} for buying assets in $S$ is defined to be the sum of money in all investment firms which contain at least one asset from $S$, i.e., 
$$ \money(S) :=  \sum_{Q \in G(S): \ Q \cap S \neq \emptyset}  {y_Q} ,$$
where ``$Q \in G(S)$'' is short for ``clique $Q$ in $G(S)$''. Strictly speaking, the summation should be over maximal cliques, but since for a non-maximal clique $Q$, $y_Q$ can be assumed to be zero, this distinction can be dropped. Notice that $\money(S)$ plays the role of {\em profit} of $S$ in a cooperative game. 
 
By Definition \ref{def.satisfied}, scenario $S$ is satisfied by imputation $y$ if $\money(S) \geq \cost(S)$ and by Definition \ref{def.core}, imputation $y$ is said to be in the {\em core} of this game if it satisfies every scenario, i.e.,
$$ \forall S \subseteq V, \ \ \money(S) \geq \cost(S) .$$

Remark \ref{rem.top-down} summarizes the definitions given above and points out an important difference from the setting of cooperative games. 
 
\begin{remark}
 	\label{rem.top-down}
 In a cooperative game, an {\em imputation} distributes the total worth of the game among agents and the {\em profit of a sub-coalition} is defined to be the sum of the profits of its agents; the latter can be viewed as a {\em bottom-up process}. Clearly these processes do not apply to our game, since it has a unique agent and the notion of a ``sub-coalition'' is replaced by that of a scenario. 
 
  The natural way of defining an imputation and the money available in a scenario can be summarized as follows: an imputation distributes the total money of the game among ``large'' objects --- the firms, which are maximal cliques --- and the money available in a scenario is defined via a {\em top-down process}, by summing the money of all cliques which intersect this scenario. 
 \end{remark}

\subsection{Limitations of this Model, and Desired Properties} 
\label{sec.limitations}

The problem of computing a maximum cost independent set in an arbitrary graph is NP-hard, even if all vertex costs are unit. This NP-hardness also indicates that the game defined above lacks structural properties that could lead to an understanding of its core. In sharp contrast, the assignment game is in $\cal P$ and its LP-relaxation supports integrality of the underlying polytope, therefore leading to a characterization of its core. Hence, the model defined above, on an arbitrary graph, is too general to be  useful. 

To be useful, the model should allow for:

\begin{enumerate}
	\item A characterization of the core; in particular, determine if the core is non-empty.
	\item Efficient computation of $T$, and $\cost(S)$ for any scenario $S$.
	\item An efficient algorithm for computing an imputation in the core. 
\end{enumerate}

In the next section, we show that restricting the game to perfect graphs gives all these properties.

\section{The Investment Management Game on Perfect Graphs}
\label{sec.perfect}

In this section, we will study a restriction of the investment management game to perfect graphs. In Section \ref{sec.def-perfect} we give the required definitions and facts from the (extensive) theory of perfect graphs. We will not credit individual papers for these facts; instead, we refer the reader to Chapter 9 of the book \cite{GLS} as well as the remarkably clear and concise exposition of this theory, presented as an ``appetizer'' by Groetschel \cite{Grotschel1999my}. In Section \ref{sec.char-core} we will use these facts to characterize the core of this game.

\subsection{Definitions and Preliminaries}
\label{sec.def-perfect}

\begin{definition}
	\label{def.omega-chi}
	Given a graph $G = (V, E)$, $\omega(G)$ denotes its {\em clique number}, i.e., the size of the largest clique in it and $\chi(G)$ denotes its {\em chromatic number}, i.e., the minimum number of colors needed for its vertices so that the two endpoints of any edge get different colors. 
\end{definition}

\begin{definition}
	\label{def.perfect}
	A graph $G = (V, E)$ is said to be {\em perfect} if and only if the clique number and chromatic number are equal for each vertex-induced subgraph of $G$, i.e., 	
	$$ \forall S \subseteq V, \	\omega(G(S)) = \chi(G(S)).$$
\end{definition}

Let $\oG$ denote the {\em complement} of $G$, i.e., $\oG = (V, \oE)$, where $\oE$ is the complement of $E$, with $\forall \ u, v \in V, (u,v) \in E$ if and only if $(u, v) \notin \oE$. A central fact about perfect graphs is that $G$ is perfect if and only if $\oG$ is perfect. 

\begin{definition}
	\label{def.stable-set}
Set $S \subseteq V$ is said to be a {\em stable set} in $G$, also sometimes called an {\em independent set}, if no two vertices of $S$ are connected by an edge, i.e., $\forall \ u, v \in S, \ (u, v) \notin E$. Let $w: V \rightarrow \QQ_+$ be a weight function on the vertices of $G$. 
%The {\em worth} of the {\em investment management game}, defined in Section \ref{sec.graph}, is the weight of a maximum weight stable set in $G$.  
\end{definition}

Let $G$ be an arbitrary graph. Clearly any clique in $G$ can intersect a stable set in at most one vertex, and therefore the constraint in LP (\ref{eq.stable-primal}) is satisfied by every stable set; note that variable $x_v$ indicates the extent to which $v$ is picked in a fractional stable set. LP (\ref{eq.stable-primal}) contains such a constraint for each clique in $G$ and is an LP-relaxation of the maximum weight stable set problem in $G$. However, LP (\ref{eq.stable-primal}) has exponentially many constraints, one corresponding to each clique in $G$; moreover, it is NP-hard to solve in general \cite{GLS}.

	\begin{maxi}
		{} {\sum_{v \in V}  {w_{v} x_{v}}}
			{\label{eq.stable-primal}}
		{}
		\addConstraint{x(Q)}{\leq 1 \quad}{\forall \ \mbox{clique $Q$ in} \ G}
		\addConstraint{x_{v}}{\geq 0}{\forall v \in V}
	\end{maxi}

The situation is salvaged in case $G$ is a perfect graph: with the help of the Lovasz theta function, one can can show that LP (\ref{eq.stable-primal}) can be solved in polynomial time using the ellipsoid algorithm \cite{GLS, Grotschel1999my}. 

Below is the dual LP, which is obtained by taking $y_Q$ to be the dual variable for the constraint of LP (\ref{eq.stable-primal}). The dual LP is solving a clique covering problem. 

 	\begin{mini}
		{} {\sum_{\mbox{clique $Q$ in} \ G}  {y_Q}} 
			{\label{eq.stable-dual}}
		{}
		\addConstraint{\sum_{Q \ni v} {y_Q}}{ \geq w_{v} \quad }{\forall v \in V}
		\addConstraint{y_{Q}}{\geq 0}{\forall \ \mbox{clique $Q$ in} \ G}
	\end{mini}

A key fact for our purpose is that the linear system of LP (\ref{eq.stable-primal}) is totally dual integral (TDI), see Definition \ref{def.TDI}, for perfect graphs \cite{GLS, Grotschel1999my}.  We provide a proof sketch below. 

By Definition \ref{def.omega-chi}, if $G$ is perfect, $\omega (G) = \chi(G)$. Next, consider the complement of $G$, namely $\overline{G}$; two vertices are adjacent in $G$ if and only if they are not adjacent in $\overline{G}$. Denote by $\alpha(G)$ and $\overline{\chi}(G)$ the size of the  largest stable set and the minimum number of disjoint cliques needed to cover all vertices of $G$, respectively. Clearly, $\omega(G) = \alpha(\overline{G})$ and $\chi(G) = \overline{\chi}(\overline{G})$. Furthermore, since the complement of a perfect graph is also perfect, we get that for a perfect graph $G$, $\alpha(G) = \overline{\chi}(G)$. 

Consider LP (\ref{eq.stable-primal}) and its integer programming formulation for the case that the weight function $w \in \{0, 1\}^n$, and let $L_p(w)$ and $I_p(w)$ denote their optimal objective function values. For the same restriction on $w$, let $L_d(w)$ and $I_d(w)$ denote the optimal objective function values of LP (\ref{eq.stable-dual}) and its integer programming formulation. Now,
 $$I_p(w) \leq L_p(w) = L_d(w) \leq I_d(w), $$ 
 where the equality follows from the LP-duality theorem, and the inequalities follow from the relation between integral and fractional solutions. 

For $w \in \{0, 1\}^n$, let $G'$ be the subgraph of $G$ induced on vertices $v$ for which $w_v = 1$. Since $G'$ is also perfect, $I_p(w) = \alpha(G')$ and $I_d(w) = \overline{\chi}(G')$. Since $\alpha(G') = \overline{\chi}(G')$, we get that $I_p(w) = I_d(w)$ and hence equality holds for all four programs defined above.   

To show that LP (\ref{eq.stable-primal}) is TDI we must show that for every weight function $w \in \ZZ_+^n$, the optimal objective function value of the dual is integral. Since $L_d(w) = I_d(w)$, this is the case if $w \in \{0, 1\}^n$. Finally, the results of Fulkerson \cite{Fulkerson1972anti} and  Lovasz \cite{Lovasz1972normal} show that this integrality implies integrality of the optimal dual even if $w \in \ZZ_+^n$. Therefore the linear system of LP (\ref{eq.stable-primal}) is TDI and hence the polytope defined by it is integral; the vertices of this polytope are stable sets.

\subsection{Characterizing the Core of the Game}
\label{sec.char-core}

Since the firms in the investment management game are {\em maximal} cliques, we first restrict the  constraint in LP(\ref{eq.stable-primal}) to maximal cliques only to obtain LP (\ref{eq.maximal-primal}). Observe that if $Q$ is a clique in $G$ and $Q'$ is a sub-clique of $Q$, then the constraint in LP (\ref{eq.stable-primal}) corresponding to $Q'$ is redundant and can be removed, since it is implied by the constraint corresponding to $Q$. Continuing in this manner, we are left with constraints corresponding to maximal cliques only. Therefore, the two LPs are equivalent.

	\begin{maxi}
		{} {\sum_{v \in V}  {w_{v} x_{v}}}
			{\label{eq.maximal-primal}}
		{}
		\addConstraint{x(Q)}{\leq 1 \quad}{\forall \ \mbox{maximal clique $Q$ in} \ G}
		\addConstraint{x_{v}}{\geq 0}{\forall v \in V}
	\end{maxi}
	
The dual of LP (\ref{eq.maximal-primal}) is given below. 

 	\begin{mini}
		{} {\sum_{\mbox{clique $Q$ in} \ G}  {y_Q}} 
			{\label{eq.maximal-dual}}
		{}
		\addConstraint{\sum_{Q \ni v} {y_Q}}{ \geq w_{v} \quad }{\forall v \in V}
		\addConstraint{y_{Q}}{\geq 0}{\forall \ \mbox{maximal clique $Q$ in} \ G}
	\end{mini}

We next show that the linear system of LP (\ref{eq.maximal-primal}) is also TDI. This is not immediate, since as stated after Definition \ref{def.TDI}, TDI is not a property of the polytope but of the particular linear system chosen to define it. 

\begin{lemma}
\label{lem.TDI}
	The linear system of LP (\ref{eq.maximal-primal}) is TDI.
\end{lemma}

\begin{proof} 
Let $Q$ be a maximal clique in $G$, $Q'$ be its sub-clique and $x$ be a vector of variables $x_v$ for each vertex $v \in V$. For any $x$ whose coordinates have been set to non-negative numbers, if $x(Q) \leq 1$ then $x(Q') \leq 1$, i.e., the latter constraint is redundant. Therefore the linear system of (\ref{eq.maximal-primal}) is obtained by removing redundant constraints from the linear system of (\ref{eq.stable-primal}). Hence the two primal LPs (\ref{eq.stable-primal}) and (\ref{eq.maximal-primal}) are equivalent.

Since the linear system of LP (\ref{eq.stable-primal}) is TDI, for any integer-valued cost vector $w$ such that the linear program (\ref{eq.stable-primal}) has an optimum solution, the dual linear program (\ref{eq.stable-dual}) has an {\em integer} optimal solution. We will use this fact to show that the same holds for LPs (\ref{eq.maximal-primal}) and (\ref{eq.maximal-dual}) as well.

Let $w$ be an integer-valued cost vector such that the linear program (\ref{eq.stable-primal}) has an optimum solution and let $y$ be the corresponding integer optimal solution to LP (\ref{eq.stable-dual}). Since the two primal LPs given above are equivalent, LP (\ref{eq.maximal-primal}) also has an optimum solution for $w$. 

Use $y$ to construct a dual solution $y'$ for LP (\ref{eq.maximal-dual}) using the following operation: Suppose $y_{Q'} > 0$, where $Q'$ is not a maximal clique. Let $Q$ be any maximal clique of which $Q'$ is a sub-clique. Now subtract $y_{Q'}$ from $y_{Q'}$ and add it to $y_{Q}$. Observe that new dual obtained is integral, its objective value remains unchanged and it is feasible for LP (\ref{eq.stable-dual}) since for each vertex $v$, ${\sum_{Q \ni v} {y_Q}}$ remains unchanged. 

Repeat this operation until $y_{Q} > 0$ only if $Q$ is a maximal clique. Call the resulting dual $y'$. Clearly, $y'$ is an integral optimal dual for LP (\ref{eq.maximal-dual}). Therefore LP (\ref{eq.maximal-primal}) is TDI. 
\end{proof}

In our setting, an optimal dual distributes the worth of the game among the maximal cliques of $G$. However, dividing the money $y_Q$, given to clique $Q$, among the vertices in the clique in not very meaningful, see Example \ref{ex.Paley} for a detailed explanation. As described in Theorem \ref{thm.stable}, in our game, the money available to a {\em any scenario} is defined via a different process, see also Remark \ref{rem.top-down}.

\begin{example}
	\label{ex.Paley}
The $3 \times 3$ Paley graph, which is a perfect graph, is shown in Figure \ref{fig.Paley}. It has three disjoint maximum stable sets of size three each, one of each color, and three disjoint maximal cliques, one of which is shown in bold. Under unit cost for each asset, the worth of the investment management game on this graph is 3 and the optimal dual assigns 1 to each of the three cliques. Consider three scenarios, each consisting of a maximum stable set. The optimal investment in each scenario is to buy all three of its assets, requiring 3 units of money. If the cliques were to distribute their unit money to the assets, then at most one of these scenarios can be satisfied: by each clique allocating its money to a different colored asset of the {\em same} scenario. Therefore, the dual on a clique cannot be distributed among its vertices. 
\end{example}

%%%%%%%%%%%%%%%%%%%%%%%%%%%%%%%%%%%%%%%%%%%%%%%%%%%%%%%%%%%

\begin{figure}[h]
\begin{center}
\includegraphics[width=3.2in]{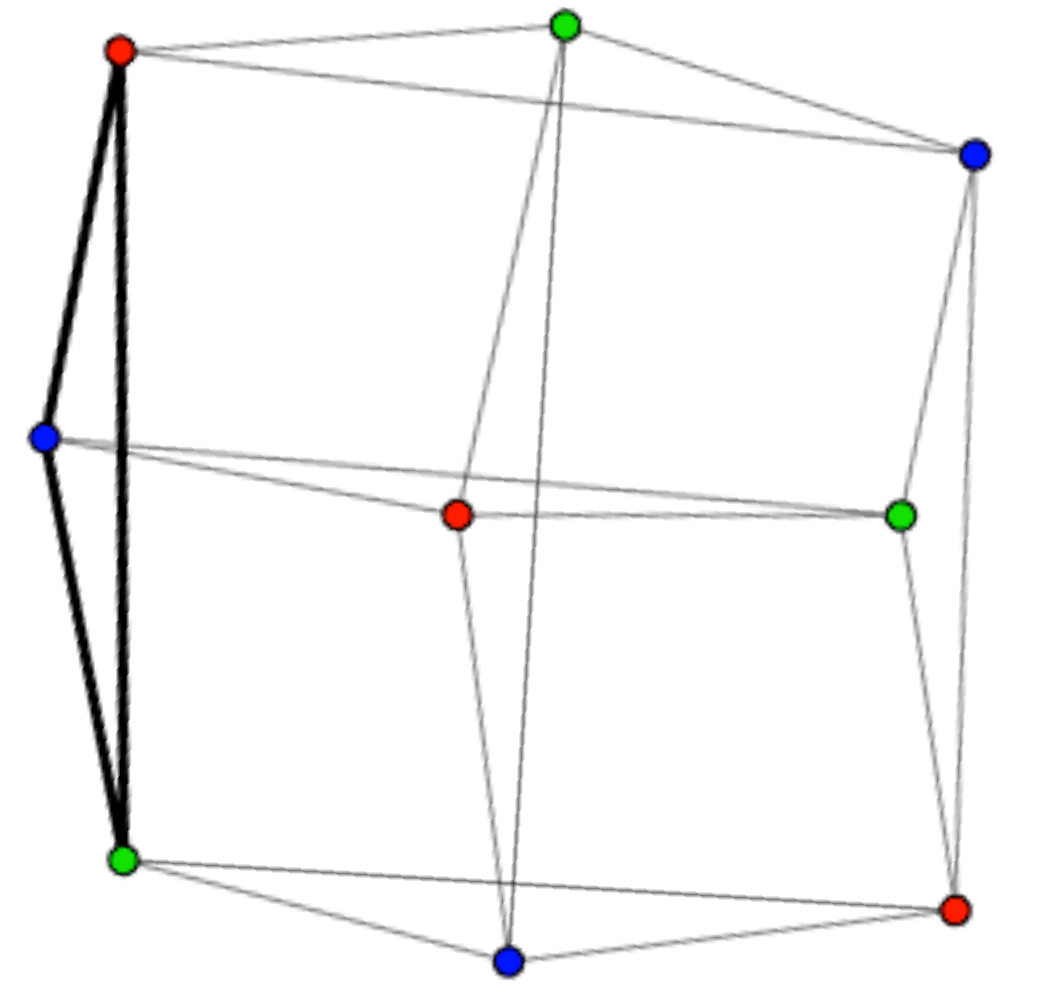}
\caption{The $3 \times 3$ Paley graph for Example \ref{ex.Paley}}
\label{fig.Paley}
\end{center}
\end{figure}

%%%%%%%%%%%%%%%%%%%%%%%%%%%%%%%%%%%%%%%%%%%%%%%%%%%%%%%%%%%%%%%

\begin{theorem}
	\label{thm.stable}
	An imputation $y$ is in the core of the stable set game over a perfect graph if and only if it is an optimal solution to the dual LP(\ref{eq.maximal-dual}). 
\end{theorem}
	
\begin{proof}
($\Leftarrow$) 
Let $y$ be an optimal solution to the dual LP(\ref{eq.maximal-dual}). By Lemma \ref{lem.TDI}, the linear system of LP(\ref{eq.maximal-primal}) is TDI and therefore a maximum cost stable set in $G$ is an optimal solution to this LP. This fact together with the LP-duality theorem give:
$$ T = \cost (V) = \sum_{Q \in G}  {y_Q} = \money(V) .$$
Therefore $y$ is an imputation.  

Consider a scenario $S \subseteq V$. By Definition \ref{def.imputation}, 
$$ \money(S) :=  \sum_{Q \in G: \ Q \cap S \neq \emptyset}  {y_Q} .$$

Clearly $Q \cap S$ is a clique in $G(S)$. Next, we will define a function, $z$, on cliques in $G(S)$ as follows: For a clique $Q'$ in $G(S)$, define
\[ z_{Q'} := \sum_{Q \in G: \ Q \cap S = Q'} {y_Q} .\]
Since each clique $Q$ of $G$ which has a non-empty intersection with $S$ will contribute to exactly one clique in $G(S)$, namely $Q \cap S$, we get 
$$ \money(S) =   \sum_{Q' \in G(S)}  {z_Q'} .$$

Next, we observe that $z$ is a feasible solution to the restriction of LP (\ref{eq.stable-dual}) to $G(S)$ because 
		$$ \forall v \in S: \ \sum_{Q' \in G(S): \ Q' \ni v} {z_{Q'}} = \sum_{Q \in G: \ Q \ni v} {y_Q} \geq w_v .$$

Since the subgraph of a perfect graph is also perfect, $G(S)$ is a perfect graph and the restriction of LP (\ref{eq.stable-primal}) to $G(S)$ satisfies TDI. Therefore a maximum cost stable set in $G(S)$ is an optimal solution to the latter LP. This fact together with weak duality give us
$$ \cost(S) \leq \sum_{Q \in G(S)}  {z_Q}  = \money(S) ,$$
i.e., imputation $y$ satisfies scenario $S$. Therefore $y$ is a core imputation. 

\bigskip

($\Rightarrow$) Next, assume that $y$ is a core imputation of the investment management game over a perfect graph $G$. 

For $v \in V$, consider the scenario $S = \{v\}$. Since $y$ is in the core, by Definition \ref{def.core},
$$ \money(S) = \sum_{Q \in G: \ Q \cap S \neq \emptyset}  {y_Q} \leq \cost(S) = w_v .$$
Therefore, $y$ is a feasible solution for LP(\ref{eq.maximal-dual}). 

By Definition \ref{def.imputation}, 
$$ \money(V) =  \sum_{Q \in  G}  {y_Q} = T = \cost(V) .$$
Therefore by the TDI of LP(\ref{eq.maximal-primal}), the objective function value of $y$ is the same as that of the optimal value of the primal. This together with the feasibility of $y$ establishes that $y$ is an optimal solution to the dual LP(\ref{eq.maximal-dual}). 
\end{proof}

\begin{corollary}
	\label{cor.stable}
	The core of the stable set game over a perfect graph is non-empty. 
\end{corollary}

\section{Discussion}
\label{sec.discussion}

As stated in the Introduction, the notion of core was given in the nineteenth century by Edgeworth \cite{Edgeworth1881mathematical} in the context of general equilibrium theory, and it was ported to cooperative game theory in the 1950s by Gillies \cite{Gillies-Core}. However, we have applied this notion in the context of a game against nature. Hence perhaps the most important question raised by our work is determining the ``natural home'' of  this notion within game theory. A related question is the following: Remark \ref{rem.contrast} draws a clear contrast between the way the core is used in cooperative game theory and in our game. Are there other ways of interpreting and using the notion of core, without taking liberties with its formal framework?

A simple way of defining the game presented in this paper would use the graph-theoretic language of stable sets and cliques. Definition \ref{def.game-management} moves a step towards modeling an economic situation. There is no doubt that our model falls short of providing a solution concept for a realistic economic situation; however, that was not the intent of this paper. We simply wanted to provide evidence that the beautiful solution concept of core is more versatile than previously envisaged. We hope other researchers will be able to build on this viewpoint to find realistic applications of the notion of core outside of cooperative game theory.

\section{Acknowledgements}
\label{sec.ack}

I wish to thank Martin Bullinger insightful comments on the writeup, and Gerard Cornuejols, Federico Echenique, Naveen Garg, Martin Groetschel, Ruta Mehta, Joseph Root, Thorben Trobst and Richard Zeckhauser for valuable discussions.

	\bibliographystyle{alpha}
	\bibliography{refs}

\end{document}